\documentclass[a4paper,11pt]{article}
\usepackage[titletoc,toc,title]{appendix}
\usepackage[scientific-notation=true]{siunitx}
\usepackage[framed,numbered,autolinebreaks,useliterate]{mcode}
\usepackage{a4}
\usepackage{euscript}
\usepackage{boxedminipage}
\usepackage{dcolumn}
\usepackage{makeidx} 
\usepackage{float}
\usepackage{fancybox}
\usepackage{rotating}
\usepackage{algorithmicx}
\usepackage{algpseudocode}
\usepackage{latexsym}
\usepackage{amssymb}
\usepackage{amsmath}
\usepackage{epsfig}
\usepackage{color}
\usepackage{verbatim}
\usepackage[normalem]{ulem} 
\usepackage{mathrsfs}

\newtheorem{alg}{Algorithm}[section]

\newtheorem{fig}{Figure}[section]
\newtheorem{tab}{Table}[section]
\newtheorem{exa}{Example}[section]
\newtheorem{prop}{Proposition}[section]
\newtheorem{lem}{Lemma}[section]
\newtheorem{rem}{Remark}[section]
\newtheorem{theorem}{Theorem}[section]
\newtheorem{cor}{Corollary}[section]
\newtheorem{lemma}[theorem]{Lemma}

\newcommand{\bc}{\begin{center}}
\newcommand{\ec}{\end{center}}

\newcommand{\beq}{\begin{equation}}
\newcommand{\eeq}{\end{equation}}
\newcommand{\eit}{\end{itemize}}
\newcommand{\bit}{\begin{itemize}}
\newcommand{\bex }{\begin{exa}}
\newcommand{\eex}{\end{exa}}
\newcommand{\betab}{\begin{tab}}
\newcommand{\epsilonb}{\end{tab}}
\newcommand{\efig}{\end{fig}}
\newcommand{\befig}{\begin{fig}}
\newcommand{\berem }{\begin{rem}}
\newcommand{\erem}{\end{rem}}
\newcommand{\balg}{\begin{alg}}
\newcommand{\ealg}{\end{alg}}

\newcommand{\eprop}{\end{prop}}
\newcommand{\beprop}{\begin{prop}}
\newcommand{\belem}{\begin{lem}}
\newcommand{\elem}{\end{lem}}
\renewcommand{\beth}{\begin{theorem}}
\newcommand{\etheor}{\end{theorem}}

\renewcommand{\epsilon}{\varepsilon}

\def\real{\hbox{\rm\setbox1=\hbox{I}\copy1\kern-.45\wd1 R}}
\def\peal{\hbox{\rm\setbox1=\hbox{I}\copy1\kern-.45\wd1 P}}
\def\eal{\hbox{\rm\setbox1=\hbox{1}\copy1\kern-.45\wd1 I}}
\def\neal{\hbox{\rm\setbox1=\hbox{I}\copy1\kern-.45\wd1 N}}

\newcommand{\ben}{\begin{enumerate}}
\newcommand{\een}{\end{enumerate}}

\newcommand{\Em}{\mathbb E}
\newcommand{\Pm}{\mathbb P}

\newcommand{\be}{\begin{equation}}
\newcommand{\ee}{\end{equation}}

\newcommand\del{\bgroup\markoverwith
    {\color{red}{\rule[0.5ex]{2pt}{0.7pt}}}\ULon}

\newif\ifnotes\notestrue
\notesfalse             

\newcommand\slavadelete{\bgroup\markoverwith
    {\color{blue}{\rule[0.5ex]{2pt}{0.7pt}}}\ULon}

\begin{document}

\begin{center}
{\large  {\bf Monte Carlo for estimating exponential convolution
}}
\end{center}

\vspace{1cm}

\begin{center}

{Ilya Gertsbakh$^{\mbox
{\footnotesize a}}$, Eyal Neuman$^{\mbox{\footnotesize b}}$, Radislav Vaisman$^{\mbox{\footnotesize b}}$}

{\small
\vspace{0.2cm}

$^{\mbox{\footnotesize a}}$ Department of Mathematics, \\
Ben Gurion University, Beer-Sheva 84105,
Israel\\
\textsf{elyager@bezeqint.net} \\

\vspace{0.1cm}

$^{\mbox{\footnotesize b}}$
Faculty of Industrial Engineering and Management,  \\
Technion, Israel Institute of Technology,
Haifa 32000, Israel \\
\textsf{eyaln@tx.technion.ac.il
}  \\
\textsf{slava@tx.technion.ac.il
}  \\
\vspace{0.2cm}
}
\end{center}

\bc
\today
\ec

\begin{abstract}
In this note we study the numerical stability problem that may 	take place when calculating the cumulative distribution function  of the {\it Hypoexponential} random variable. This computation is extensively used during the execution of Monte Carlo network reliability estimation algorithms. In spite of the fact that analytical formulas  are available, they can be unstable in practice.  This instability occurs frequently when estimating very small failure probabilities $(10^{-30}-10^{-40})$ that can happen for example while estimating the unreliability of telecommunication systems. In order to address this problem, we propose a simple unbiased estimation algorithm that is capable of handling a large number of variables. We show that the proposed estimator has a bounded relative error and that it compares favorably with other existing methods.
\end{abstract}



\section{Introduction}
Network Reliability problem appears in many real life applications such as transportation, social and computer networks, communication, and many more. One approach to handle this problem is by using a Monte Carlo (MC) technique. Some MC methods require computation of the {\it Cumulative Distribution Function} (CDF) of the {\it Hypoexponential} random variable.

We can state the reliability problem as follows. Suppose we are given an undirected graph $G(V,E,T)$  where $V$ and $E$ are the vertex and edge sets respectively and let $T\subseteq V$ be some terminal set of nodes. Suppose also that edges are subject to failure and for all $e \in E$ there is a corresponding failure probability $q_e$. Under this setting we can ask for the probability that the terminal set $T$ be connected. We call the latter an {\it UP} state.

One of the well-studied approaches to handle this problem is called an {\it Evolution Monte Carlo } (EMC) method \cite{Gertsbakh}.
The main idea is that at time zero no edges are present in the graph. Next, we assign each edge a corresponding exponential random variable that represents the time this edge is ``born". Naturally, there is a time when the network enters the {\it UP}  state. The {\it Evolution Monte Carlo} method studies those times and delivers  the corresponding network reliability (for details see \cite{Gertsbakh}). While executing the EMC algorithm, we need to perform many calculations of the  form  $\Pm(\sum_{i=0}^n{X_i}\leq t)$.
Note that $X_i \sim exp(\lambda_i)$, so this sum is distributed {\it Hypoexponentially} and the corresponding complementary CDF can be computed using a matrix exponential:
\beq\label{conv1}
 \Pm(\sum_{i=0}^n{X_i}\geq t)=e_1e^{Dt}\textbf{1}=e_1\sum_{k=0}^{\infty}{\frac{D^kt^k}{k!} \textbf{1}}
\eeq
where $e_1=(1,0,\cdots,0)$ is a $1 \times n$ vector, \textbf{1} is a $n \times 1$
column vector of ones, and
\[ D =
 \begin{pmatrix}
  -\lambda_1 & \lambda_1 & 0 & \cdots & 0 \\
  0 & -\lambda_2 & \lambda_2 &  \cdots & 0 \\
  \vdots  & \vdots  & \vdots  & \ddots & \vdots  \\
  0  & \cdots  & 0  & -\lambda_{n-1} & \lambda_{n-1}  \\
  0  & \cdots  & 0  & 0 & \lambda_{n}  \\
 \end{pmatrix} \]
is a $n \times n$ matrix \cite{BotevReliability}. For the rest of this section we concentrate on the methods used to perform this computation.\\\\
First, we examine the exact algorithms available.
\bit
\item  If $\lambda_1 > \lambda_2 > \cdots > \lambda_n$ is satisfied, formula (\ref{conv1}) can be written as
\beq\label{exact.eq}
    \Pm(\sum_{i=0}^n{X_i}\leq t)=1-\sum_{i=0}^n{e^{-\lambda_it}\prod_{j \neq i} {\frac{\lambda_j}{\lambda_j-\lambda_i}} }
\eeq
and computed in $O(n^2)$ time following Ross in \cite{ross}. Unfortunately, it was noted that this formula suffers from numerical instability. For example, consider the following $\lambda$ values.
\[
\begin{array}{c}
  \lambda_1=10.00, \lambda_2=9.99,\lambda_3=9.98,\lambda_4=9.97, \lambda_5=9.96, \lambda_6=9.95, \\ \lambda_7=9.94\lambda_8=9.93,\lambda_9=9.92,	 \lambda_{10}=9.91,\lambda_{11}=9.9,	\lambda_{12}=9.89.
\end{array}
\]
Using a MatLab code proposed in \cite{zdravko.handbook} to calculate $\Pm(\sum_{i=0}^{12}{X_i}\leq 1)$ we observe that this probability is equal to $-134,217,727$. The result can be verified using the {\it convolution1} code in Appendix \ref{exact.algorithm}.\\

\item A much better approach was tested by Botev et al. \cite{BotevReliability} and exploited a new matrix exponential algorithm called {\it scaling and squaring} that was introduced by Higham in \cite{Higham05thescaling}.  The {\it convolution2} MatLab implementation is attached in Appendix \ref{exact.algorithm}. This method is very stable but more expensive in the sense of CPU time when compared to {\it convolution1}.
\eit
Next, we introduce randomized methods that output the estimation of the desired value.
\bit
\item The {\it Cross Entropy} {\it (CE)} method is a powerful technique for solving difficult estimation
and optimization problems, based on Kullback-Leibler (or cross-entropy) minimization \cite{Kroese2011}.
This method was pioneered by Rubinstein in 1999 \cite{rub97} and is based on  an adaptive importance sampling procedure for the estimation of rare-event probabilities.

\item The $Splitting$ method is another common technique to deal with counting, combinatorial optimization
and rare-event estimation, but unlike the {\it CE} method that is based on Importance Sampling, the $Splitting$ procedure relies on the Markov Chain Monte Carlo (MCMC) approach.  $Splitting$ dates back to Kahn and Harris \cite{kahn51} and Rosenbluth and Rosenbluth \cite{rosenbluth55}. The main idea  is to partition the
state-space of a system into a series of nested subsets and to consider the rare event as the intersection of a nested sequence of events.

\item The Conditional Monte Carlo Algorithm {\it (G-S)} proposed by Gertsbakh and Shpungin in \cite{Gertsbakh}, Section $7.3$, p. $91$. The main idea of this approach is to sample the
exponential random variables recursively while avoiding rare-event settings.
This technique was especially designed to handle the numerical problems that may occur during the exponential convolution calculation.
\eit
The rest of the note is organized as follows.
In section \ref{IS.sec} we introduce our algorithm and prove that  it is unbiased and has a bounded relative error.
In section \ref{numerical.sec} we present numerical results and show that our approach can be compared with other methods.
Finally, section \ref{conclusion.sec} presents some concluding remarks.

\section{{\it IS} Algorithm}\label{IS.sec}
Given independent exponential random variables $X_1,\dots ,X_n$ such that $X_i\sim exp(\lambda_i)$, we propose
to sample from different densities and use likelihood ratios respectively. The details are presented in the following algorithm.\\

\balg \textbf{{\it IS} Algorithm}\label{alg.IS}\\
{\it Input:} $\lambda_1,\cdots,\lambda_n$\\
{\it Output:} $\widehat{\Pm}\big(\sum_{i=1}^n{X_i} \leq 1\big)$
\begin{algorithmic}[1]
\State $res \gets 0$
    \For{$i = 1 \to N$}
        \State Sample $y_1,\cdots,y_n$, such that $y_i \sim exp(n)$
        \If{$\sum_{i=1}^n{y_i} \leq 1$}
            \State $res \gets res + \frac{\prod_{i=1}^n{\lambda_i e^{-\lambda_iy_i}}}{\prod_{i=1}^n{n e^{-ny_i}}}$
         \EndIf
    \EndFor
\State return $\frac{res}{N}$
\end{algorithmic}
\ealg
Let us define
\begin{equation}\label{Prob}
    \ell=\Pm\big(\sum_{i=1}^n{X_i} \leq 1 \big).
\end{equation}
Note that the algorithm outputs an estimator to $\Em(Z)$, where
\begin{equation}\label{Def-Z}
    Z = 1_{\{\sum_{i=1}^n{Y_i}\leq 1\}}\frac{\prod_{i=1}^n{\lambda_i e^{-\lambda_i Y_i}}}{\prod_{i=1}^n{n e^{-nY_i}}}.
\end{equation}
For a formal proof that $\Em[Z]$ is an unbiased estimator of $\Pm\big(\sum_{i=1}^n{X_i} \leq 1 \big)$ see Lemma \ref{unbiased}.\\\\
The following corollary immediately follows  from the definition of a relative error and from Theorem \ref{variance.bound}.
\begin{cor}\label{re.bound}
The relative error of the {\it IS} Algorithm satisfies
\beq
RE\leq \sqrt{ \frac{\sqrt{n}e^{2(\bar\lambda-\underline{\lambda})+1}}{N}}
 \eeq
 where $n$ is a number of exponential random variables in the sum, $\bar{\lambda}=\max_{i=1,...,n}\{\lambda_i\}$,   $\underline{\lambda}=\min_{i=1,...,n}\{\lambda_i\}$, and $N$ is the sample size.
\end{cor}


\section{Numerical Results}\label{numerical.sec}
We conducted many numerical experiments using all the algorithms mentioned earlier. In general, we came to the conclusion that for most practical purposes, the exact algorithm {\it convolution2} should be preferred. Unfortunately, when rare event settings are involved the latter may fail. In this section we consider the performance of the proposed algorithms on $3$ models. We performed all computations on an Intel Core i5 laptop with 4GB RAM.
We use the same algorithm parameters for all models.
\begin{itemize}
    \item {\it IS}:  $N=100  n$ sample size
    \item {\it Cross Entropy}:  $\rho=0.3$, $\alpha=0.5$ and $N=100n$ sample size both for parameter estimation and final sampling
    \item {\it Splitting}:  $\rho=0.3$ and $N=1,000$ sample size
    \item {\it G-S}:  $N=100,000$ sample size
    \item The relative error ($\widehat{RE}$) calculation is based on $K=10$ independent runs.The $\widehat{RE}$ was calculated as
     \beq  \widehat{RE} = {S\over \widetilde \ell},
 \label{re}
 \eeq
where $$\widehat \ell = \widehat \Pm\big(\sum_{i=1}^n{X_i} \leq 1 \big), \ S^2 = {1\over K-1}\sum_{i=1}^K (\widehat \ell_i - \widetilde \ell)^2 \rm \ and  \
\widetilde \ell = {1\over K}\sum_{i=1}^K \widehat \ell_i.$$

\item $\widehat{RTV}$ - relative time variance is used to compare different algorithms; it is defined as the simulation time in seconds multiplied by the squared relative error.

\end{itemize}
We consider the following models.
\bit
\item \textbf{Model 1:}  $\sum_{i=1}^{10}{X_i}$ where $X_i\sim exp(\lambda)$ are i.i.d  exponential random variables with $\lambda = 0.03$.
\item \textbf{Model 2:}  $\sum_{i=1}^{10}{X_i}$ where $X_i\sim exp(\lambda)$ are i.i.d exponential  random variables with $\lambda = 0.01$.
\item \textbf{Model 3:}  $\sum_{i=1}^{10}{X_i}$ where $X_i\sim exp(\lambda_i)$. The corresponding $\lambda$ values are given below.
\[
\begin{array}{c}
  \lambda_1=0.01, \lambda_2=0.011,\lambda_3=0.009,\lambda_4=0.01, \lambda_5=0.011, \\
  \lambda_6= 0.009,  \lambda_7=0.01 \lambda_8=0.011,\lambda_9=0.009,	 \lambda_{10}=0.01.
\end{array}
\]

\eit
The following tables summarize our results.
\begin{table}[H]
\caption{Average performance of $10$ runs of the algorithms for \textbf{Model 1}}
\bc \begin{tabular}{c|c|c|c|c}
\multicolumn{1}{c}{} & \multicolumn{1}{c}{}  & \multicolumn{1}{c}{} & \multicolumn{1}{c}{} & \multicolumn{1}{c}{}\\
$Algorithm$ &  {$\widehat{\Pm}(\sum{X_i} \leq 1)$}  & $\widehat{RE}$ & $\widehat{RTV}$& CPU \\ \hline \hline
{\it IS}	&	\num{1.61E-22}	&	\num{5.98E-02}	&	\num{3.38E-04}	&	 0.094	\\	\hline
{\it Cross Entropy}	&	\num{1.58E-22}	&	\num{9.24E-02}	&	\num{1.90E-03}	 &	0.222	\\	\hline
{\it Splitting}	&	\num{1.32E-22}	&	\num{5.53E-01}	&	1.37	&	4.501	 \\	\hline
{\it G-S}	&	\num{1.44E-22}	&	\num{3.94E-01}	&	\num{2.85E-02}	&	 0.184	\\	\hline
\end{tabular}
\ec \label{tab.11}
\end{table}
The exact {\it convolution2} algorithm  delivers $\Pm(\sum{X_i} \leq 1)=\num{1.1102e-016}$ as an output. Unfortunately, Algorithm {\it convolution1} cannot be used for equal $\lambda$ values.\\\\
\begin{table}[H]
\caption{Average performance of $10$ runs of the algorithms for \textbf{Model 2}}
\bc \begin{tabular}{c|c|c|c|c}
\multicolumn{1}{c}{} & \multicolumn{1}{c}{}  & \multicolumn{1}{c}{} & \multicolumn{1}{c}{} & \multicolumn{1}{c}{}\\
$Algorithm$ &  {$\widehat{\Pm}(\sum{X_i} \leq 1)$}  & $\widehat{RE}$ & $\widehat{RTV}$& CPU \\ \hline \hline
{\it IS}	&	\num{2.75E-27}	&	\num{6.05E-02}	&	\num{3.50E-04}	&	 0.096	\\	\hline
{\it Cross Entropy}	&	\num{2.74E-27}	&	\num{6.73E-02}	&	\num{5.82E-04}	 &	0.128	\\	\hline
{\it Splitting}	&	\num{2.97E-27}	&	\num{5.55E-01}	&	1.74	&	5.667	 \\	\hline
{\it G-S}	&	\num{2.35E-27}	&	\num{3.77E-01}	&	\num{2.57E-02}	&	 0.181	\\	\hline
\end{tabular}
\ec \label{tab.12}
\end{table}
The exact {\it convolution2} algorithm delivers $\Pm(\sum{X_i} \leq 1)=\num{1.1102e-016}$ as an output. Note that the algorithm outputs the same value for both \textbf{Model 1} and \textbf{Model 2}.\\\\
\begin{table}[H]
\caption{Average performance of $10$ runs of the algorithms for \textbf{Model 3}}
\bc \begin{tabular}{c|c|c|c|c}
\multicolumn{1}{c}{} & \multicolumn{1}{c}{}  & \multicolumn{1}{c}{} & \multicolumn{1}{c}{} & \multicolumn{1}{c}{}\\
$Algorithm$ &  {$\widehat{\Pm}(\sum{X_i} \leq 1)$}  & $\widehat{RE}$ & $\widehat{RTV}$& CPU \\ \hline \hline
{\it IS}	&	\num{2.56E-27}	&	\num{3.19E-02}	&	\num{9.75E-05}	&	 0.096	\\	\hline
{\it Cross Entropy}	&	\num{2.60E-27}	&	\num{3.42E-02}	&	\num{1.61E-04}	 &	0.138	\\	\hline
{\it Splitting}	&	\num{2.35E-27}	&	\num{4.18E-01}	&	\num{9.75E-01}	&	 5.588	\\	\hline
{\it G-S}	&	\num{2.14E-27}	&	\num{2.47E-01}	&	\num{1.10E-02}	&	 0.180	\\	\hline
\end{tabular}
\ec \label{tab.13}
\end{table}
Algorithm {\it convolution1} cannot deliver a meaningful answer and  {\it convolution2} algorithm delivers $\Pm(\sum{X_i} \leq 1)=\num{-2.2204e-016}$ as an output. Note that in this case the stability is lost and the algorithm outputs  $\Pm(\sum{X_i} \geq 1)>1$.
\section{Conclusions}\label{conclusion.sec}
In this note, we developed a new importance sampling algorithm for computing the CDF of the {\it Hypoexponential} random variable. We proved that the proposed estimator is efficient and its performance compares favorably with other existing methods. Based on our numerical results we conclude that
in situations with no rare events involved, one should prefer to use the exact {\it convolution2} method that is still relatively fast and very stable. Naturally, when the exact method fails, which may happen in case of very small probabilities, one should apply some Monte Carlo approximation. MCMC based $Splitting$ is too slow to be used in reliability applications. The {\it G-S}  has a good performance and very easy to implement but it seems that its relative error is inferior when compared to {\it IS}. The {\it Cross Entropy} and the proposed {\it IS} algorithm are comparable, but {\it IS} is much simpler to implement.\\\\\\\\
\textbf{ACKNOWLEDGMENT}\\\\
We are thoroughly grateful to anonymous reviewers for their valuable constructive  remarks and suggestions.
\newpage
\bibliographystyle{plain}
\bibliography {convolution-estimation-bib}

\newpage

\begin{appendices}

\section{Proofs}\label{proofs}

\begin{lemma}\label{unbiased}
The output of {\it IS} Algorithm \ref{alg.IS} is unbiased.
\end{lemma}

\begin{proof}
Let   $Y_i \sim \exp(n)$ , $i=1,2,...,n$ be independent exponentially distributed random variables. Recall that we are looking for an unbiased estimator
of $P(X_1+X_2+...+X_n \leq 1)$, where $X_i$ are independent and $X_i \sim \exp(\lambda_i)$.  Our estimator is

\begin{equation}
Z=1_{\{\sum_{i=1}^n Y_i \leq 1\}}\frac{\prod_{i=1}^n \lambda_i e^{-\lambda_i Y_i}}{\prod_{i=1}^n n e^{-n Y_i}}.
\end{equation}
Note that the joint density function of $\textbf{Y}=(Y_1,Y_2,...,Y_n)$ is $\Psi(\textbf{v})=\prod_{i=1}^n n e^{-n v_i}$.
Now
\[
E[Z]=\int  (n) \int_{v_i \geq0, v_1+...+v_n \leq 1}\frac{\prod_{i=1}^n \lambda_i e^{-\lambda_i v_i}}{\prod_{i=1}^n n e^{-n v_i}}\prod_{i=1}^n n e^{-n v_i}dv_1dv_2...dv_n=
\]
\[
\int  (n) \int_{v_i \geq0, v_1+...+v_n \leq 1}\prod_{i=1}^n \lambda_i e^{-\lambda_i v_i}dv_1dv_2...dv_n=P(X_1+X_2+...+X_n \leq 1). \,\,Q.E.D.
\]
\end{proof}


\begin{theorem}\label{variance.bound}
Let $Z$ be defined as in (\ref{Def-Z}). Then we have,
\begin{equation}\label{Eff-Bound}
    \frac{\mathbb{E}(Z^2)}{(\mathbb{E}(Z))^2}\leq \sqrt{n}e^{2(\bar\lambda-\underline{\lambda})+1},
\end{equation}
where $\bar{\lambda}=\max_{i=1,...,n}\{\lambda_i\}$ and  $\underline{\lambda}=\min_{i=1,...,n}\{\lambda_i\}$.
\end{theorem}
\begin{proof}
Denote by
\begin{equation}\label{Y-not}
    Y:=\sum_{i=1}^nY_i.
\end{equation}
By the definition of the random variables $y_i, \ i=1,...,n$, we have that $Y$ is distributed $\textrm{Erlang}(n,n)$ and therefore it has the following probability density
\begin{equation}\label{Density}
    f_{Y}(y)=\frac{n^{n}}{(n-1)!}y^{n-1}e^{-ny}, \ \ y> 0.
\end{equation}
Define
\begin{equation}\label{I}
    I(n,x):=\int_{0}^{x}t^{n-1}e^{t}dt.
\end{equation}
From (\ref{I}),(\ref{Def-Z}),(\ref{Y-not}) and (\ref{Density}) we have,
\begin{eqnarray} \label{upper-bound}
  \mathbb{E}(Z^2) &= &  \bigg(\frac{\prod_{i=1}^n{\lambda_i}}{\prod_{i=1}^nn}\bigg)^2\mathbb{E}\bigg(1_{\{\sum_{i=1}^n{Y_i}\leq 1\}}\frac{\prod_{i=1}^n{e^{-2\lambda_iY_i}}}{\prod_{i=1}^n{ e^{-2nY_i}}}\bigg)   \\
  &\leq &  \bigg(\frac{\prod_{i=1}^n{\lambda_i}}{\prod_{i=1}^nn}\bigg)^2\mathbb{E}\bigg(1_{\{\sum_{i=1}^n{Y_i}\leq 1\}}\frac{{e^{-2\underline{\lambda}\sum_{i=1}^nY_i}}}{{ e^{-2n\sum_{i=1}^n Y_i}}}\bigg) \nonumber  \\
  &= &  \bigg(\frac{\prod_{i=1}^n{\lambda_i}}{\prod_{i=1}^nn}\bigg)^2 \mathbb{E}\bigg(1_{\{Y\leq 1\}}e^{2(n-\underline{\lambda})Y}\bigg) \nonumber  \\
  &=&\bigg(\frac{\prod_{i=1}^n{\lambda_i}}{\prod_{i=1}^nn}\bigg)^2\frac{n^{n}}{(n-1)!}\int_{0}^{1}x^{n-1}e^{(n-2\underline{\lambda})x}dx\nonumber  \\
  &=&\bigg(\frac{\prod_{i=1}^n{\lambda_i}}{\prod_{i=1}^nn}\bigg)^2\frac{n^n}{(n-1)!}
\frac{1}{(n-2\underline{\lambda})^{n}}\int_{0}^{n-2\underline{\lambda}}x^{n-1}e^{x}dx \nonumber \\
&=&\bigg(\frac{\prod_{i=1}^n{\lambda_i}}{\prod_{i=1}^nn}\bigg)^2\frac{n^n}{(n-1)!}
\frac{1}{(n-2\underline{\lambda})^{n}}I(n,n-2\underline{\lambda}). \nonumber
\end{eqnarray}
Recall the definition of the \emph{lower incomplete gamma function},
\begin{equation}\label{L-I-g}
    \gamma(n,x):=\int_{0}^{x}t^{n-1}e^{-t}dt.
\end{equation}
Use (\ref{L-I-g}),(\ref{Def-Z}),(\ref{Y-not}) and (\ref{Density}) to get,
\begin{eqnarray} \label{Lower-bound}
\mathbb{E}(Z) &=& \frac{\prod_{i=1}^n\lambda_i}{\prod_{i=1}^n n}\mathbb{E}\bigg(1_{\{\sum_{i=1}^n{Y_i}\leq 1\}}\frac{{e^{-\sum_{i=1}^n{\lambda_i}Y_i}}}{{ e^{-n\sum_{i=1}^n{Y_i}}}}\bigg)   \\
&\geq&  \frac{\prod_{i=1}^n\lambda_i}{\prod_{i=1}^n n}\mathbb{E}\bigg(1_{\{\sum_{i=1}^n{Y_i}\leq 1\}}\frac{{e^{-\bar{\lambda}\sum_{i=1}^nY_i}}}{{ e^{-n\sum_{i=1}^n{Y_i}}}}\bigg) \nonumber  \\
&=&  \frac{\prod_{i=1}^n\lambda_i}{\prod_{i=1}^n n}\mathbb{E}\bigg(1_{\{Y\leq 1\}}{e^{(n-\bar{\lambda})Y}}\bigg) \nonumber  \\
&=&  \frac{\prod_{i=1}^n\lambda_i}{\prod_{i=1}^n n}\cdot\frac{n^{n}}{(n-1)!}\int_{0}^{1}x^{n-1}e^{-\bar{\lambda}x}dx \nonumber  \\
&=&  \frac{\prod_{i=1}^n\lambda_i}{\prod_{i=1}^n n}\cdot\frac{n^{n}}{(n-1)!}\frac{1}{\bar\lambda^n}\int_{0}^{\bar{\lambda}}x^{n-1}e^{-x}dx \nonumber  \\
&=&  \frac{\prod_{i=1}^n\lambda_i}{\prod_{i=1}^n n}\cdot\frac{n^{n}}{(n-1)!}\frac{1}{\bar\lambda^n}\gamma(n,\bar{\lambda}). \nonumber
\end{eqnarray}
From (\ref{upper-bound}) and (\ref{Lower-bound}) we get
\begin{eqnarray} \label{eff-1}
\frac{\mathbb{E}(Z^2)}{(\mathbb{E}(Z))^2}\ &\leq & \frac{(n-1)!}{n^{n}}
\frac{\bar{\lambda}^{2n}}{(n-2\underline{\lambda})^n}\frac{I(n,n-2\underline{\lambda})}{(\gamma(n,\bar{\lambda}))^2}.
\end{eqnarray}
By a simple calculation we obtain the following bounds on the functions $I$ and $\gamma$,
\begin{equation}\label{I-Bound}
    I(n,x)\leq \frac{x^{n}e^{n}}{n}, \ \forall x \in [0,\infty), \ n\in \mathbb{N},
\end{equation}
\begin{eqnarray} \label{Gamma-Bound}
\gamma(n,x)\geq \frac{x^n}{ne^{x}}, \ \forall x \in [0,\infty), \ n\in \mathbb{N}.
\end{eqnarray}
Recall Stirling's formula
\begin{eqnarray} \label{stir}
n! \leq n^{n+1/2}e^{-n+1}, \ \forall \ n\in \mathbb{N}.
\end{eqnarray}
Apply (\ref{I-Bound})--(\ref{stir}) on (\ref{eff-1}) to get
\begin{eqnarray} \label{eff-2}
\frac{\mathbb{E}(Z^2)}{(\mathbb{E}(Z))^2}\ &\leq & \frac{e^{-n+1}}{\sqrt{n}}
\frac{\bar{\lambda}^{2n}}{(n-2\underline{\lambda})^n}\frac{I(n,n-2\underline{\lambda})}{(\gamma(n,\bar{\lambda}))^2} \nonumber \\
&\leq&\frac{e^{-n+1}}{\sqrt{n}}
\frac{\bar{\lambda}^{2n}}{(n-2\underline{\lambda})^n}
\frac{\frac{1}{n}(n-2\underline{\lambda})^{n}e^{n-2\underline{\lambda}}}{\frac{\bar{\lambda}^{2n}}{n^2e^{2\bar\lambda}}} \nonumber \\
&=&\sqrt{n}e^{2(\bar\lambda-\underline{\lambda})+1},
\end{eqnarray}
and we get (\ref{Eff-Bound}).
\end{proof}

\section{MatLab code for exact computation}\label{exact.algorithm}
\begin{lstlisting}
    function ell=convolution1(t,nu)
        % computes P(A_1+...+A_b>t) exactly,
        % where A_i distributed Exp(nu(i)) independently;
        % nu has to be decreasing (sorted) sequence
        b=length(nu); % parameters of the waiting times
        w=zeros(b,b); % b is critical number
        w(l,l)=l;
        for k=l:b-l
            for j=l:k
                w(k+l,j)=w(k,j)*nu(b-k)/(nu(b-k)-nu(b-j+l));
                w(k+l,k+l)=l-sum(w(k+l,1 :k));
            end
        end
        ell=w(b,:)*exp(-nu(end:-l: 1)’*t); % probability
    end
\end{lstlisting}

\begin{lstlisting}
    function ell=convolution2(t,nu)
       % computes P(A_1+...+A_b>t) exactly,
       % where A_i ~ Exp(nu(i)) independently;
       % nu has to be decreasing (sorted) sequence
       b=length(nu); % parameters of the waiting times
       A=diag(-nu)+diag(nu(1:b-1),1);
       A=expm(A*t);
       ell=sum(A(1,:));
    end
\end{lstlisting}

\end{appendices}

\end{document}